\documentclass[12pt]{article} 
\usepackage[sectionbib]{natbib}
\usepackage{array,epsfig,fancyheadings,rotating}
\usepackage[colorlinks=true,citecolor=blue]{hyperref}  
\usepackage{afterpage}
\usepackage{sectsty, secdot}
\sectionfont{\fontsize{12}{14pt plus.8pt minus .6pt}\selectfont}
\renewcommand{\theequation}{\thesection\arabic{equation}}
\subsectionfont{\fontsize{12}{14pt plus.8pt minus .6pt}\selectfont}

\usepackage{amsmath}
\usepackage{amssymb}
\usepackage{amsfonts}
\usepackage{multirow}
\usepackage{amsthm}

\newtheorem{theorem}{Theorem}
\newtheorem{lemma}{Lemma}
\newtheorem{condition}{Condition}

\theoremstyle{definition}

\def\hatEps{\hat{\mathcal{E}}}
\def\Eps{{\mathcal{E}}}

\def\cov{{\text{cov}}}

\begin{document}


\renewcommand{\baselinestretch}{2}

\markright{ \hbox{\footnotesize\rm  
}\hfill\\[-30pt]
\hbox{\footnotesize\rm Statistica Sinica
}\hfill }

\markboth{\hfill{\footnotesize\rm Zhendong Huang AND Davide Ferrari} \hfill}
{\hfill {\footnotesize\rm Fast Construction of Composite Likelihoods} \hfill}

\renewcommand{\thefootnote}{}
$\ $\par


\fontsize{12}{14pt plus.8pt minus .6pt}\selectfont \vspace{0.8pc}
\centerline{\large\bf Fast Construction of Optimal Composite Likelihoods}
\vspace{.4cm} 
\centerline{Zhendong Huang$^{\dagger}$ and Davide Ferrari$^{\ast}$} 
\vspace{.4cm} 
\centerline{\it $^{\dagger}$University of Melbourne, $^{\ast}$University of Bolzano }
 \vspace{.55cm} \fontsize{9}{11.5pt plus.8pt minus.6pt}\selectfont


\begin{quotation}
\noindent {\it Abstract:}
A composite likelihood is a combination of low-dimensional likelihood objects useful  in applications where the data have complex structure.   Although composite likelihood construction is a crucial aspect influencing both computing and statistical properties of the resulting estimator, currently there does not seem to exist a universal rule to combine low-dimensional likelihood objects that is statistically justified and fast in execution.  This paper develops a methodology  to select and combine the most informative low-dimensional likelihoods from a large set of candidates while carrying out parameter estimation. The new procedure minimizes the distance between composite likelihood and  full likelihood scores subject to a constraint representing  the afforded computing
cost. The selected composite likelihood is sparse in the sense that it contains a relatively  small number of informative sub-likelihoods while the noisy terms are dropped.  The resulting estimator is found to have asymptotic variance close to that of the  minimum-variance estimator constructed using all the low-dimensional likelihoods.

\vspace{9pt}
\par
\end{quotation}\par

\def\thefigure{\arabic{figure}}
\def\thetable{\arabic{table}}

\renewcommand{\theequation}{\thesection.\arabic{equation}}

\fontsize{12}{14pt plus.8pt minus .6pt}\selectfont

\afterpage{\chead[]{Fast construction of composite likelihoods}}

\section{Introduction}\label{sec:introduction}

The likelihood function is central to many statistical analyses. There are a number of situations, however, where the  full likelihood is computationally intractable or difficult to specify. These challenges have motivated the development of composite likelihood methods, which avoid 
intractable full likelihoods by combining a set
of low-dimensional likelihood objects. Due
to its flexible framework and computational advantages, composite likelihood inference has become popular
in many areas of statistics; e.g., 
see \cite{Varin&al11} for an overview and applications.

Let $X \subseteq \mathbb R^d$ be a $d \times 1$ vector random variable with density in the family $f(x;\theta)$, where $\theta \in \Theta \subseteq \mathbb{R}^p$ is an unknown parameter. Let $\theta^\ast$ denote the true parameter. Suppose now that the full $d$-dimensional distribution is difficult to specify or compute, but it is possible to specify $m$ tractable  pdfs
$f_1(s_1;\theta), \dots, f_m(s_m; \theta)$ for
sub-vectors $S_1,\dots, S_m$ of $X$,  each with dimension
smaller than $d$. For example, 
$S_1$ may represent a single element like $X_1$, a variable pair
like $(X_1 , X_2)$ or a conditional sub-vector like $X_1|X_2$. The total number of sub-models $m$ may grow quickly with  $d$; for example, taking all variable pairs in $X$ yields $m=d(d-1)/2$. Thus, from one vector $X$ we
form  the composite log-likelihood 
\begin{align*}
\ell(\theta, w; X) =  \sum_{j=1}^m w_j \ell_j(\theta; X) = \sum_{j=1}^m w_j \log f_{j}(S_j; \theta),
\end{align*}
where $w$ is a $m \times 1$ vector of constants to be  determined as a solution to an optimality problem. For $n$ independent, identically distributed vectors $X^{(1)}, \dots, X^{(n)}$ we define 
\begin{align*}
\ell(\theta, w; X^{(1)}, \dots, X^{(n)}) =  \sum_{i=1}^n \ell(\theta;w, X^{(i)}).
\end{align*}
The score functions are obtained in the usual way:
\begin{align*}  
U(\theta, w; X)  = \nabla \ell(\theta, w; X) = \sum_{j=1}^m w_{j} U_{j}(\theta; X), \ \ \ \ \ U_j (\theta; X) = \nabla \ell_{j}(\theta ; X),\\
U(\theta, w;  X^{(1)}, \dots, X^{(n)})  = \nabla \ell(\theta, w; X^{(1)}, \dots, X^{(n)})   = \sum_{j=1}^m w_{j} \sum_{i=1}^n U_{j}(\theta; X^{(i)}), 
\end{align*}
where ``$\nabla$" denotes the gradient with respect to $\theta$. The maximum composite likelihood estimator $\hat \theta(w)$ is defined as the solution to the    estimating equation
\begin{align} \label{eq:est_equation}
U(\theta, w;  X^{(1)}, \dots, X^{(n)})  = \sum_{j=1}^m w_{j} \sum_{i=1}^n U_{j}(\theta; X^{(i)}) = 0, 
\end{align}
for some appropriate choice of $w$. Besides computational advantages and modeling
flexibility, one reason for the popularity of the composite likelihood estimator is that  it enjoys  
properties analogous to maximum likelihood  \citep{Lindsay88, Lindsay&al11, Varin&al11}. Under typical regularity conditions, the composite likelihood estimator is asymptotically normal with mean $\theta^\ast$ and variance $\{ {G}(\theta^\ast,w)\}^{-1}$, where  
\begin{equation}\label{eq:Godambe}
{G}(\theta^\ast, w) =  H(\theta^\ast, w)  \{ K(\theta^\ast, w)  \}^{-1} H(\theta^\ast, w) ,
\end{equation}
is the so-called Godambe information matrix,  $H(\theta, w) = - E\{\nabla U(\theta, w; X)\}$ and $K(\theta, w)= 
\cov\{U(\theta, w ;X)\}$ are the $p\times p$  sensitivity  and variability matrices, respectively. Although the maximum composite likelihood estimator is consistent,  $  G(\theta^\ast, w)$ is generally different from the
Fisher information $ \cov \{\nabla  \log f(X; \theta^\ast)\}$, with the two coinciding only in special cases where  $H(\theta^\ast, w) = K(\theta^\ast, w)$.

The choice of $w$ determines both the statistical properties and computational efficiency
of the composite likelihood estimator \citep{Lindsay&al11, xu2011robustness,huang2020specification}. On one hand, the established theory of unbiased
estimating equations would suggest to find $w$ so to maximize
$tr\{ {G}(\theta^\ast, w)\}$ \citep[Chapter 2]{heyde2008quasi}. Although
theoretically appealing, these optimal weights depend on inversion of the score covariance matrix whose estimates are often singular. Different selection strategies to balance the trade-off between statistical efficiency and computing cost have been explored in the literature.  A common practice is to retain all feasible
sub-likelihoods with $w_j =1$, for all $j\geq 1$, but this
is undesirable from either computational parsimony or statistical efficiency
viewpoints, since the presence of too many correlated scores inflates the variability matrix $K$ \citep{Cox&Reid04, ferrari2016parsimonious}.  A smaller subset may be selected by setting some of the $w_j$s equal to zero, but determining a suitable subset remains challenging. \cite{dillon2010stochastic} and \cite{ferrari2016parsimonious} develop stochastic approaches where sub-likelihoods are sampled according to a statistical efficiency criterion. Ad-hoc methods have been developed depending on the type of model under exam; for example, in spatial data analysis it is often convenient to consider  sub-likelihoods corresponding to close-by observations; e.g., see \cite{heagerty1998composite, sang2014tapered, bevilacqua2015comparing}.

Motivated by this gap in the literature, the present paper develops a methodology to select sparse composite likelihoods in large problems by retaining only the most informative scores in the estimating equations (\ref{eq:est_equation}), while dropping the   noisy ones. To this end, we propose to minimize the distance between the maximum likelihood score and the composite likelihood score subject to a constraint representing  the overall computing cost.

The reminder of the paper is organized as follows. In Section \ref{sec:method} we describe the main sub-likelihood selection and combination methodology. In Section \ref{sec:properties}, we discuss the properties of our method. Particularly, while Theorem \ref{prop:proposition2} shows that the proposed empirical composition rule is asymptotically equivalent to the optimal composition rule that uses all the available scores,   Theorems \ref{thm:consistent} and \ref{thm:normality} give consistency and asymptotic normality of the resulting parameter estimator. In Section \ref{sec:example}, we discuss some illustrative examples related to common families of models. In Section \ref{sec:realdata}, we apply our method to  real Covid-19 epidemiological data. Finally, in Section \ref{sec:discussion} we conclude and provide final remarks.

\afterpage{\chead[]{}}

\section{Main methodology} \label{sec:method}

\subsection{Penalized score distance minimization}

We propose to solve equation (\ref{eq:est_equation}) with weights $w=w_\lambda(\theta)$ selected by minimizing the penalized score distance
\begin{equation} \label{eq:criterion}
\dfrac{1}{2}E \left\|   U^{ML} (\theta; X)- U(\theta, w; X) 
\right\|^2_2   + \lambda  \sum_{j=1}^m  \left\vert   w_j \right\vert,
\end{equation}
where $U^{ML}(\theta;x)= \nabla \log f(x; \theta)$ is the maximum likelihood score, $\Vert \cdot \Vert_2$ denotes the $L_2$-norm, $\lambda 
\geq 0$ is a regularization parameter. The resulting minimizer, say $w_{\lambda}(\theta)$, is then used for parameter estimation by solving the composite likelihood estimating equation (\ref{eq:est_equation}) in $\theta$ with $w = w_{\lambda}(\theta)$.

The vector of coefficients minimizing  (\ref{eq:criterion}) is allowed to contain positive, negative  or zero values, although negative elements do not cause any specific concerns in our method. The size of such coefficients is expected to be larger for those sub-likelihoods that are strongly correlated with the full likelihood. Thus, a negative coefficient associated with certain sub-likelihood score does  not imply that such a sub-likelihood is less informative, but simply means that it has negative correlation with the maximum likelihood score.

The  first term  in the objective (\ref{eq:criterion})  aims at improving statistical efficiency  by
finding a composite score close to the maximum likelihood score. Note that minimizing the first term alone (when $\lambda =0$) corresponds to finding finite-sample optimal $O_F$-optimal estimating equations. This criterion formalizes the idea of minimization of the variance in estimating equations; see \cite[Ch.1]{heyde2008quasi}.  \cite{Lindsay&al11} point out that this type of criterion is suitable in the context of composite likelihood estimation, although finding a general computational procedure to minimize such a criterion in large probelms is still an open problem.

The  term $\lambda  \sum_{j=1}^m \left\vert   w_j
\right\vert$ is a  penalty discouraging 
overly complex scores.  The
geometric properties of the $L_1$-norm  penalty ensure that   several elements in the solution  
$w_{\lambda}(\theta)$ are zero for sufficiently large $\lambda$, thus simplifying the resulting estimating equations. This is a key property of the 
proposed approach which is exploited to reduce the computation burden. 

The optimal solution $w_\lambda(\theta)$ may be interpreted as one that maximizes 
statistical accuracy, subject to a given level of afforded computing.  Alternatively, $w_\lambda(\theta)$ may be viewed as a composition rule that minimizes the 
likelihood complexity subject to some afforded efficiency loss
compared to maximum likelihood. The  constant $\lambda$ balances the trade-off between   statistical efficiency and computational cost:  $\lambda=0$ is optimal in terms of asymptotic efficiency, but offers no reduction in likelihood complexity, while for increasing $\lambda>0$  informative data subsets and their scores might be tossed away. 

Difficulties related to the direct minimization of (\ref{eq:criterion})  are the presence 
of the intractable likelihood score function $U^{ML}$  and the expectation depending on  
the unknown parameter $\theta$. Up to an additive term independent of $w$, the penalized score distance in (\ref{eq:criterion}) can be expressed as
\begin{align} \label{eq:criterion1}
\dfrac{1}{2} E \left\|   
U(\theta, w; X)
\right\|^2_2 -  E \left\{   
U^{ML}(\theta; X)^{\top} U(\theta, w; X) \right\} + \lambda  \sum_{j=1}^m  \left\vert   w_j \right\vert.
\end{align}
Let $M(\theta; X)=(U_1(\theta; X), \dots, U_m(\theta; X))$ be the $p \times m$ matrix with columns given by $p \times 1$ score vectors $U_j(\theta; X)$ $(j=1,\dots,m)$. Then, the first term of (\ref{eq:criterion1}) may be expressed as  $w^\top J(\theta) w/2$, where $J(\theta)$ is the $m \times m$ score covariance matrix 
$$
J(\theta) = E\{ M(\theta; X)^\top M(\theta; X) \}.
$$ 
Note that at $\theta=\theta^\ast$, assuming unbiased scores with $E\{U_{j}(\theta^\ast;X)\} = 0$ ($j=1,\dots,m$), the second Bartlett equality gives $E\{U^{ML}(\theta^\ast;X) U_{j}(\theta^\ast;X)^\top \}=   E\{U_{j}(\theta^\ast;X) U_{j}(\theta^\ast;X)^\top\} = - E \{\nabla U_{j}(\theta^\ast;X) \}$; this implies that the second term in (\ref{eq:criterion1}) is $
- w^\top \text{diag}\{ J(\theta^\ast) \}
$, 
where $\text{diag}(A)$ denotes the diagonal vector of the matrix $A$. Therefore, (\ref{eq:criterion1})  may be approximated by
\begin{align} \label{eq:criterion2}
d_\lambda (\theta, w)  = 
\dfrac{1}{2} w^\top J(\theta) w  -  
w^\top \text{diag}\{ J(\theta) \} + \lambda  \sum_{j=1}^m  \left\vert   w_j \right\vert.
\end{align}

For $n$ independent observations  $X^{(1)}, \dots, X^{(n)}$ on $X$,    we  obtain  the empirical composition rule $\hat w_\lambda(\theta)$ by minimizing the empirical criterion 
\begin{align}\label{eq:criterion_empirical}
\hat d_\lambda(\theta, w) =    \dfrac{1}{2}w^\top \hat J( \theta)  
w - w^\top \text{diag}\{\hat J(\theta) \}  +   \lambda  \sum_{j=1}^m   \left\vert   w_j \right\vert,
\end{align}
where 
$
\hat J(\theta) = n^{-1} \sum_{i=1}^n M(\theta; X^{(i)})^\top  M(\theta; X^{(i)})$.

The final composite likelihood estimator may be found by  replacing $w = \hat w_\lambda(\theta)$ in  (\ref{eq:est_equation}) and then solving the following estimating equation with respect to $\theta$
\begin{equation} \label{eq:est_eq3}
U(\theta, \hat w_\lambda(\theta);  X^{(1)}, \dots, X^{(n)}) = 0.
\end{equation}
Although $\hat w_\lambda(\theta)$ is generally smooth in a neighborhood of $\theta^\ast$, it may exibit a number of nondifferentiable points on the parameter space $\Theta$. This means that convergence of standard gradient-based root-finding algorithms, such as the Newton-Raphson algorithm, is not guaranteed.

To address this issue, we propose to take a preliminary root-$n$ consistent estimate $\tilde \theta$, and find the final estimator $\hat \theta_{\lambda}$ instead solving the estimating equation 
\begin{equation} \label{eq:est_eq3}
U(\theta, \hat w_\lambda(\tilde \theta);  X^{(1)}, \dots, X^{(n)}) = 0.
\end{equation}
where $ \hat w_\lambda(\tilde \theta)$ is a quantity fully dependent on the data. A preliminary estimate is often easy to obtain, for example by solving (\ref{eq:est_equation}) with $w_j=1$ for all $1 \le j \le m$. Alternatively, a computationally cheap root-$n$ consistent estimate may be obtained by setting $w_j = 1$ for $j$ in some random subset $S \subseteq \{1, \dots, m\}$ and $w_j=0$ otherwise.

\subsection{Computational aspects and selection of $\lambda$} \label{Sec:lambda}

For the numerical examples in this paper, the following implementation is considered. We first compute the sparse composition rule $\hat w_{\lambda}(\tilde \theta)$, by minimizing the convex criterion (\ref{eq:criterion_empirical}) with $\theta = \tilde \theta$  being a preliminary root-$n$ consistent estimator.  Minimization of (\ref{eq:criterion_empirical}) is implemented through the least-angle regression  algorithm  \citep{Efron04}. Finally, the estimator $\hat \theta_{\lambda}$ is obtained by a one-step Newton-Raphson update starting from $\theta = \tilde \theta$  applied to (\ref{eq:est_equation}). See Chapter 5 in \cite{van2000asymptotic} for an introduction to one-step estimation. As preliminary estimator $\tilde \theta$, one may choose the composite likelihood estimator with uniform composition rule $w_j =1$, for all $j\geq 1$. From theoretical view point, any root-$n$ consistent initial estimator $\tilde \theta$ leads to the same asymptotic results of the final estimate. We find that the impact of the initial estimates is negligible in many situations. 

Analogous to  the least-angle algorithm originally developed by \cite{Efron04}   in the context of sparse linear regression, each step of our implementation includes the score  $U_j(\theta; X)$ having the  largest correlation with the  residual difference 
$U_j( \theta ; X) - U( \theta , w; X)$, followed by an adjustment step on $w$. An alternative computing approach is to solve  (\ref{eq:criterion_empirical}) with respect to $\theta$ with $w= w_\lambda(\theta)$ using a Newton-Raphson algorithm with  $w_\lambda(\theta)$ updated through  a coordinate-descent approach as in \cite{wu2008coordinate} in each iteration.

Selection of $\lambda$ is an aspect of practical importance since it balances the trade-off between statistical and computational efficiency. In many practical applications this choice ultimately depends on the available computing resources and the objective of one's analysis. Although    a universal approach for selection is not sought in this paper, the following heuristic strategy may be considered. Taking a   grid  $\Lambda$ of  values for $\lambda$ corresponding to different numbers of selected scores, we consider $\hat{\lambda}=\max \{   \lambda \in \Lambda: \phi(\lambda) > \tau \}$, for some user-specified  constant $0 <\tau \leq 1$, where $\phi(\lambda) =   \text{tr}\{ \hat J_{\lambda} \}/\text{tr}\{\hat J\}$.
Here $\hat J_{\lambda}$ denotes the empirical covariance matrix for the selected partial scores evaluated at $\tilde \theta$, whilst $\hat J$ is the covariance for all scores. Thus  $\phi(\lambda)$ can be viewed as the approximate proportion of score variance explained by the selected scores.  In practice,  one may  choose $\tau$  to be a sufficiently large  value, such as $\tau = 0.75$ or $\tau = 0.90$.

When the covariance for all scores $\hat J$ is difficult to obtain due to excessive computational burden, we propose to use an upper bound of $\phi(\lambda)$ instead. Let $\tilde \lambda \in \Lambda$ be the next value for $\lambda \in \Lambda$ smaller than $\hat \lambda$, i.e. we set $\tilde \lambda = \hat \lambda$ if $\hat\lambda = \min\{\Lambda\}$, and $\tilde \lambda = \max\{\lambda \in \Lambda: \lambda< \hat \lambda\}$ otherwise.  Note that $\phi(\hat \lambda) < \text{tr}\{ \hat J_{\hat \lambda} \}/\text{tr}\{\hat J_{\tilde \lambda}\}$, where the right hand side represents the relative proportion of score variance explained by reducing $\lambda$ from $\hat \lambda$ to $\tilde \lambda$. Thus, in practice, one may take $\hat \lambda$ such that $\text{tr}\{ \hat J_{\hat \lambda} \}/\text{tr}\{\hat J_{\tilde \lambda}\}>\delta$ for some relative tolerance level $0<\delta<1$.

\section{Properties} \label{sec:properties}

\subsection{Conditions for uniqueness} 

This section gives an explicit expression for the minimizer of the penalized score distance criterion and provides sufficient conditions for its uniqueness. The main requirement for uniqueness is that each partial score cannot be fully determined by a linear combination of other scores. Specifically, we require the following condition:

\begin{condition} \label{cond1}
	Define $U_j = U_j(\theta, X)$. For any $\lambda>0$ and $\theta \in \Theta$, the random vectors $(U_1^\top,U_1^\top U_1\pm\lambda),\dots,$ $(U_m^\top,U_m^\top U_m\pm\lambda)$ are linearly independent. 
\end{condition}
  
We note that the condition is automatically satisfied, unless some partial score is perfectly correlated to the others, which is rarely the case in real applications.

For a vector $a \in \mathbb{R}^m$, we use $a_{\mathcal{E}}$ to denote the sub-vector corresponding to index $\mathcal{E} \subseteq \{1, \dots, m\}$, while $A_{\mathcal{E}}$
denotes the sub-matrix of the squared matrix $A$ formed by taking rows and columns corresponding to $\mathcal{E}$. 
The notation $\text{sign}(w)$ is used for the vector sign function with $j$th element taking values $-1$,  $0$ and $1$ if $w_j<0$, $w_j= 0$ and $w_j>0$. Let $\eta = 0$ if $\hat J(\theta)$ is positive definite, or else $\eta = \max_{x\in {\mathbb R^q}}[  \text{diag}\{ \hat J(\theta)\}^\top V(\theta)x / \| V(\theta) x \|_1 ]$, where $q$ denotes the number of zero eigenvalues of $\hat J(\theta)$ and $V(\theta)$ is a $m\times q$ matrix collecting eigenvectors corresponding to zero eigenvalues.

\begin{theorem} \label{prop:uniqueness}
	Under Condition \ref{cond1}, for any $\theta \in \Theta$ and $\lambda > \eta$, the minimizer  of the penalized distance $\hat d_\lambda(\theta, w)$ defined in (\ref{eq:criterion_empirical}) is unique with probability one and given by 
	\begin{equation*}  
	\hat w_{\lambda, \hatEps}(\theta) = \left\{ \hat J_{\hatEps}(\theta) \right\}^{-1} 
	\left[\text{diag}\left\{ \hat J_{\hatEps}(\theta) \right\} - \lambda \ \text{sign}\{\hat w_{\lambda, \hatEps}(\theta)\}  \right],  \  \  \hat w_{\lambda,  \setminus \hatEps}(\theta) = 0,
	\end{equation*}
	where $\hatEps \subseteq \{1, \dots, m\}$ is the index set defined as
	\begin{align}\label{eq:epsilonhat}
	\hatEps = \Bigg\{ j   :   \Bigg\vert & n^{-1} \sum_{i=1}^n  U_j(\theta;  X^{(i)})^\top      U_j(\theta;  X^{(i)})   \notag \\  
	&- n^{-1} \sum_{i=1}^n  U_j(\theta; X^{(i)})^\top U(\theta, \hat w_\lambda(\theta); X^{(i)})   \Bigg\vert \geq   \lambda \Bigg\},
	\end{align}
	and $\setminus \hatEps$ denotes the complement index set $\{1, \dots, m\} \setminus \hatEps$. Moreover, $\hat w_{\lambda, \hatEps}(\theta)$ contains at most $np \wedge m$ non-zero elements.
\end{theorem}

When $\lambda= 0$, we have $\hatEps = \{1, \dots ,m \}$, meaning that the corresponding composition rule $\hat w_{\lambda, \hatEps}(\theta)$ does not contain zero elements.  In this case, the empirical score covariance matrix  $\hat J(\theta)$ is required to be non-singular which is certainly violated when $np < m$. Even for the case $np > m$,  $\hat J(\theta)$ may be singular due to the presence of largely correlated  partial  scores.   On  the  other  hand,  setting $\lambda > \eta$  always  gives  a  non-singular score covariance matrix and guarantees existence of $\hat w_{\lambda, \hat \Eps}(\theta)$. For sufficiently large $\lambda$, a relatively small subset of scores is selected. The formula in (\ref{eq:epsilonhat}) suggests that the $j$th score is selected when it contributes enough information  in the overall composite likelihood. Particularly,  the $j$th score is included if the estimated absolute difference between its  Fisher information   $E\{U_j(\theta; X)^\top  U_j(\theta; X) \}$ and the covariance with the overall composite likelihood score $E\{U_j(\theta; X)^\top  U(\theta, w; X) \}$ is greater than $\lambda$.

\subsection{Asymptotic optimality of the empirical composition rule}

The asymptotic behavior of the sparse composition rule $\hat w_\lambda(\theta)$ is investigated in this section as the sample size $n$ diverges. The main result is the convergence of  $\hat w_\lambda(\theta)$ to the ideal composition rule $w_\lambda(\theta)$, which is defined  as the minimizer of the population criterion $d_\lambda (\theta, w)$ specified  in (\ref{eq:criterion2}). Letting $\lambda \rightarrow 0$ as $n$ increases implies that the sparse rule $\hat w_\lambda(\theta)$ is asymptotically equivalent to the optimal rule $w_0(\theta)$ in terms of the estimator's variance, with the  latter, however, involving all $m$ scores. To show this,   an additional technical requirement on the  covariance between sub-likelihood scores is introduced.

\begin{condition} \label{cond2}
	 For all $j,k \geq 1$, $sup_{\theta \in \Theta} |  \hat J(\theta)_{jk} - J(\theta)_{jk} | \to 0 $ in probability as $n \to \infty$, where  $\hat J(\theta)_{jk}$ and $J(\theta)_{jk}$ are the $\{j,k\}$th element of $\hat J(\theta)$ and $ J(\theta)$ respectively.	  Moreover, each element of $J(\theta)$ is continuous and bounded, and the smallest eigenvalue of $J(\theta)$ is bounded away from zero for all $\theta \in \Theta$ uniformly.   
\end{condition}

\begin{theorem} \label{prop:proposition2}
	Under Conditions \ref{cond1} and \ref{cond2}, for any $\lambda>0$ and $\theta \in \Theta$,  we have $\sup_{\theta \in \Theta} \Vert \hat w_\lambda(\theta) - w_\lambda(\theta) \Vert_1  \rightarrow 0$ in probability, as $n \rightarrow \infty$.
\end{theorem}

Since the preliminary estimate $\tilde \theta$ is consistent, continuity of $w_\lambda(\theta)$ (shown in  Lemma \ref{lemma2} in the Appendix) implies immediately that  $\hat w_\lambda(\tilde \theta)$ converges to $w_\lambda(\theta^\ast)$, i.e. the empirical composition rule converges to the ideal composition rule evaluated at the true parameter.  Theorem 2 implies that the proposed sparse composite likelihood score is a suitable approximation for the optimal score involving $m$ sub-likelihoods. Specifically, under regularity conditions, we have
\begin{equation} \label{eq:equivalence}
\sup_{\theta \in \Theta} \left\Vert \frac{1}{n}\sum_{i=1}^n U(\theta, \hat w_\lambda(\tilde \theta); X^{(i)}) - \frac{1}{n}\sum_{i=1}^n U(\theta, w_0(\theta^\ast); X^{(i)}) \right\Vert_1 \rightarrow 0,  
\end{equation}
in probability as $n \rightarrow \infty$ and $\lambda \rightarrow 0$. Note  the  optimal composition rule $w_0(\theta^\ast)$ and the related Godambe information matrix $ G\{\theta, w_0(\theta^\ast)\}$   are typically hard to compute due to inverting the $m \times m$ score covariance matrix with entries $E[U_j(\theta; X)^\top U_k(\theta; X)]$, $1\leq j,k \leq m$.  On the other hand, the sparse composition rule $\hat w_\lambda(\tilde  \theta)$ and the implied Godambe information have the advantage to be computationally tractable since they only involve a fraction of scores.

\subsection{Limit behavior of the estimator $\hat \theta_\lambda$ and standard errors} \label{Sec:standarderror}

The final estimator $\hat \theta_\lambda$ is an M-estimator solving estimating equations of the form
\begin{equation}\label{equation:est_eq2}
\frac{1}{n}\sum_{i=1}^n U\{\theta, \hat w_\lambda(\tilde \theta); X^{(i)}\} = 0. 
\end{equation}
Since the vector $\hat w_\lambda(\tilde \theta)$ converges to $w_\lambda^\ast = w_\lambda(\theta^\ast)$ by Theorem \ref{prop:proposition2} and Lemma \ref{lemma2},  the limit of (\ref{equation:est_eq2}) may be written as
\begin{equation}\label{equation:est_eq2_lim}
E \left\{ U(\theta,   w_\lambda^\ast; X)\right\} = 0. 
\end{equation}

To show  that $\hat \theta_\lambda$ is consistent for the solution of (\ref{equation:est_eq2_lim}), we assume additional regularity conditions, particularly, to ensure a unique root of the ideal composite likelihood score and stochastic equicontinuity on each sub-likelihood scores.

\begin{condition} \label{cond3}
	 For all constants $c>0$, $\inf_{\{ \theta: \| \theta-\theta^\ast \|_1\geq c\} } \| E\{U(\theta,w_\lambda^\ast;X)\} \|_1 >0 =\| E\{U(\theta^\ast,w_\lambda^\ast;X)\}\|_1 $. Moreover, assume $\sup_{\theta \in \Theta} \|\sum_{i=1}^n U_j(\theta;X^{(i)})/n - E U_j(\theta; X)   \|_1 \to 0$ in probability as $n \to \infty$, for all $1 \le j \le m$.
\end{condition}

\begin{theorem}  \label{thm:consistent}
Under   Conditions \ref{cond1}-- \ref{cond3},  $\hat \theta_\lambda$ converges in probability to $\theta^\ast$.
\end{theorem}


To obtain asymptotic normality of the final estimator $\hat \theta_\lambda$, we assume the following condition for the sub-likelihood scores.

\begin{condition} \label{cond4}
Assume for all $j\geq 1$, $\mathrm{var}\{ U_j(\theta^\ast; X) \} < \infty$. In a neighborhood of $\theta^\ast$, each $U_j(\theta; x)$ is twice continuously differentiable in $\theta$, and the partial derivatives are dominated by some fixed integrable functions only depending on $x$. Moreover, assume $H(\theta^\ast,w_\lambda^\ast)$ defined in (\ref{eq:Godambe}) is nonsingular. 
\end{condition}

\begin{theorem} \label{thm:normality}
	Under Conditions \ref{cond1}--\ref{cond4}, we have
\begin{equation}\label{equation:asymptotics}
n^{1/2}  {G}_{\lambda }(\theta^\ast)^{1/2} \left( \hat \theta_\lambda - \theta^\ast \right) \rightarrow N_p(0, I_p),
\end{equation}
as $n \rightarrow \infty$, 
where $ {G}_\lambda(\theta) =   {G}(\theta, w^\ast_\lambda)$ is the $p \times p$ Godambe information matrix defined in (\ref{eq:Godambe}).
\end{theorem}

The $p \times p$  Godambe information matrix $G_\lambda(\theta)$ in (\ref{equation:asymptotics}) can be estimated   by the sandwich estimator
$\hat G_\lambda  = \hat H_\lambda \hat K_\lambda ^{-1} \hat H_\lambda$, where $p \times p$ matrices $\hat H_\lambda$ and $\hat K_\lambda$ are obtained, respectively, by replacing $\theta = \hat \theta_\lambda$ in 
\begin{align*}
\hat H_\lambda (\theta)   &= -\dfrac{1}{n}\sum_{i=1}^n \nabla U(\theta, \hat w_\lambda(\tilde \theta); X^{(i)} ),  \\
\hat K_\lambda (\theta)   &= \dfrac{1}{n}\sum_{i=1}^n U(\theta, \hat w_\lambda(\tilde \theta); X^{(i)}) U(\theta, \hat w_\lambda(\theta); X^{(i)} )^\top.
\end{align*}
Practical advantages of using the sparse composition rule $\hat w_\lambda(\theta)$ are  the reduction of computational cost and increased stability of the standard error calculations. Although the score variance matrix $\hat K_\lambda(\theta)$ may be difficult to obtain when $\lambda = 0$ due to potentially $O(m^2)$ covariance terms, choosing a sufficiently large value for $\lambda>0$ avoids this situation by reducing the number of terms in the composite likelihood score.

\section{Examples}\label{sec:example} 

\subsection{Common location for heterogeneous variates}  \label{sec:anaexample1}

Let $X \sim N_m(\theta 1_m, \Sigma)$, where the $m \times m$ covariance matrix $\Sigma$ has off-diagonal elements $\sigma_{jk}$ ($j\neq k$) and diagonal elements $\sigma^2_k$ ($j=k$). Computing the maximum likelihood estimator of $\theta$  requires   $\Sigma^{-1}$ and usually $\Sigma$ is estimated by the sample covariance $\hat{\Sigma}$. When $n<m$, however, the maximum likelihood estimator is not available since the sample covariance $\hat{\Sigma}$ is singular; on the other hand, the composite likelihood estimator is still feasible. The $j$th marginal score is  $U_j(\theta; x)=(x_j-\theta)/ \sigma_j^2$ and the composite likelihood estimating equation   based on the sample $X^{(1)}, \dots, X^{(n)}$ is
\begin{equation} \label{eq:est_eq_ex}
0 =  \sum_{j=1}^m   \dfrac{w_j}{\sigma_j^{2} } \sum_{i=1}^n (X^{(i)}_j-\theta).
\end{equation}
Then the population and empirical score covariances are $m \times m$ matrices with $jk$th entries 
$$
J(\theta)_{jk} = E\left\{ \dfrac{(X_j - \theta)(X_k - \theta)}{\sigma^2_j \sigma^2_k}\right\},  \ \ \ \  \hat J(\theta)_{jk} = \dfrac{1}{n} \sum_{i=1}^n\left\{ \dfrac{(X^{(i)}_j - \theta)(X^{(i)}_k - \theta)}{\sigma^2_j \sigma^2_k}\right\},
$$
respectively. 

It is useful to inspect the special case where $X$ has independent components ($\sigma_{jk}=0$ for all $j\neq k$). This setting corresponds to the fixed-effect meta-analysis model where estimators from  $m$ independent studies are combined to improve accuracy. From Theorem \ref{prop:uniqueness}, the optimal    composition rule is
$$
\hat w_{\lambda,j}(\theta) =  \left\{1- \dfrac{\lambda n \sigma_j^4}{\sum_{i=1}^n(X_j - \theta)^2}\right\}I\left\{ \dfrac{\sum_{i=1}^n(X^{(i)}_j - \theta)^2}{   n\sigma_j^4} > \lambda \right\}, 
$$
whilst the optimal population composition rule $w_{\lambda,j}(\theta)$ is the same 
as the above expression with sample averages replaced by expectations. The   composition rule $w_{\lambda,j}(\theta)$ evaluated at the true parameter   is $
w^\ast_{\lambda,j}= (1- \lambda \sigma^2_j ) I(\sigma^{-2}_j > \lambda)$ $(j=1,\dots,m)$. This highlights that overly noisy data subsets with variance $\sigma^{2}_j > \lambda^{-1}$ are dropped and thus do not influence the final estimator for $\theta$. Particularly, the number of  non-zero elements in $w^\ast_\lambda$ is $\sum_{j=1}^m I (\sigma^2_j< \lambda^{-1} )$.  Finally, substituting  $w_j =  \hat w_{\lambda,j}(\theta)$ in (\ref{eq:est_eq_ex})  gives the following fixed-point equation
\begin{equation}\label{eq:estimator_example1}
\theta = \left\{\sum_{j=1}^m  \dfrac{\hat w_{\lambda,j}(\theta)}{\sigma_j^{2}}\bar{X}_j \right\}/ \left\{\sum_{k=1}^m \dfrac{\hat w_{\lambda,k}(\theta)}{\sigma_k^{2}}\right\},
\end{equation}
which is a weighted average of marginal sample means $\bar{X}_j=n^{-1}\sum_{i=1}^n X_j^{(i)}$  $(j=1,\dots,m)$. The final composite likelihood estimator $\hat \theta_\lambda$ may be obtained by solving equation (\ref{eq:estimator_example1}).

When $\lambda = 0$, we have uniform weights $w^\ast_0=(1,\dots,1)^{\top}$ and the corresponding estimator $\hat \theta_0$ is the usual optimal meta-analysis solution. Although the implied estimator $\hat \theta_0$  has minimum variance, it offers no control for the overall computational cost since all $m$ sub-scores are selected. On the other hand, choosing judiciously $\lambda>0$ may lead to low computational burden with negligible efficiency loss for the resulting estimator. For instance, assuming $\sigma^2_j=j^2$,  a calculation shows 
\begin{align} \label{eq:msd}
\dfrac{1}{2} E \left\{ U(\theta , w^\ast_\lambda; X)- U(\theta , w^\ast_0; X) \right\}^2
\leq   \lambda^2  \sum_{j \in \Eps} j^2    +  \sum_{j \notin \Eps} {j^{-2}}  ,
\end{align}
where $\Eps = \{ j: j^2 < \lambda^{-1} \}$ and $\theta$ here is the true parameter. Since the number of selected scores is $\sum_{j=1}^m I\left( j^2< \lambda^{-1}\right) =  \lfloor \lambda^{-\frac{1}{2}} \rfloor$, we can write
$\lambda^2  \sum_{j \in \Eps} j^2 \leq \lambda^2\lambda^{-1}\lambda^{-\frac{1}{2}} =\lambda^{\frac{1}{2}}$, which converges to zero as $\lambda \rightarrow 0$; additional calculations also show  $\sum_{j \notin \Eps} {j^{-2}}  \rightarrow 0$  as $\lambda \rightarrow 0$. Hence, the left hand side of (\ref{eq:msd}) goes to zero as long as $\lambda \rightarrow 0$. This suggests that  the truncated composite likelihood score approximates suitably the optimal score, while  containing  a relatively small number of terms. If the elements of  $X$ are correlated with $\sigma_{jk}\neq 0$ for $j\neq k$, the partial scores contain overlapping information on $\theta$. In this case, tossing away  highly correlated partial scores would improve computing while maintaining satisfactory statistical efficiency for the final estimator.  

Figure \ref{fig:anaexample1} shows the solution path of $w^\ast_\lambda$; that is, the trajectory of the elements of the optimal composition rule $w^\ast_\lambda$ for different values of $\lambda$ and the asymptotic relative efficiency of the corresponding composite likelihood estimator $\hat \theta_\lambda$   compared to the maximum likelihood estimator for different values of $\lambda$. When $m$ is large ($m=1000$), the asymptotic relative efficiency drops gradually until only a few scores are left.  This  example illustrates that   relatively high efficiency can be reached by the selected composite likelihood equations, when a few partial scores  already contain the majority of information about $\theta$. In such cases, the final estimator $\hat \theta_{\lambda}$ with a sparse composition rule is expected to achieve a good trade-off between computational cost and statistical efficiency.

\begin{figure}[btp]
	\centering
	\quad\quad\quad $\rho=0$, $m=20$ \quad\quad\quad\quad$\rho=0.5$, $m=20$ \quad\quad\quad\quad $\rho=0.5$, $m=1000$\\
	\quad \quad \quad  Number of sub-likelihoods   \quad \quad \quad 
	\includegraphics[scale=0.65]{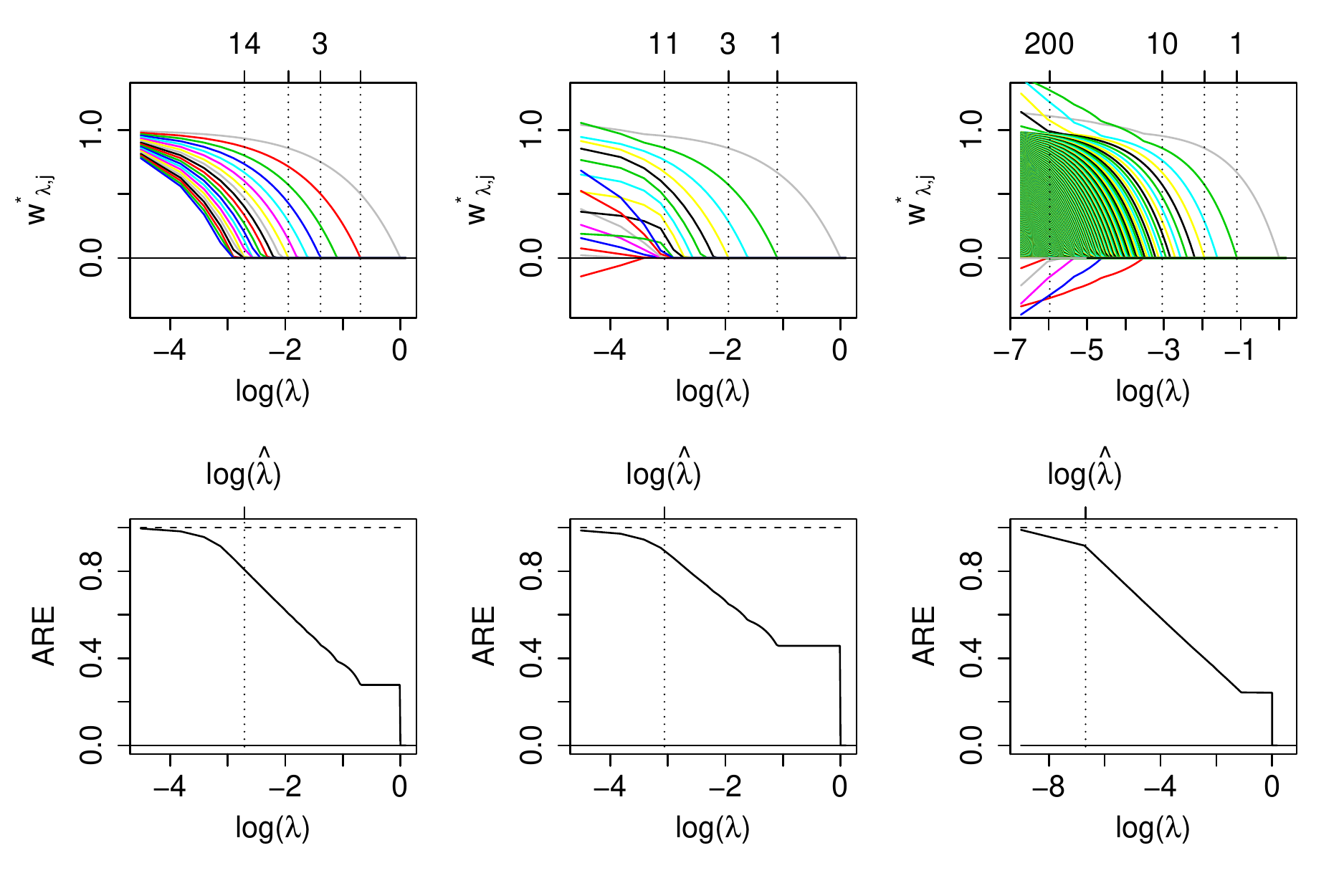}
	\caption{Top Row: Solution paths for the minimizer $w^\ast_\lambda$ of   $d_\lambda(\theta ,w)$ defined in (\ref{eq:criterion2})  at the true parameter for different values of $\lambda$ with corresponding number of sub-likelihoods shown in the top axis. 
		Bottom Row: Asymptotic relative efficiency (ARE) of $\hat \theta_\lambda$ compared to the maximum likelihood estimator.  The vertical dashed lines represent  $\hat\lambda$ selected as described in \S\ref{Sec:lambda} using $\tau=0.9$. Results correspond to the location model $X\sim N_m(\theta 1_m, \Sigma)$ with  $\Sigma_{jk} = j$ ($j=k$) and $\Sigma_{jk} =\rho (jk)^{1/2}$ ($j \neq k$). }
	\label{fig:anaexample1}
\end{figure}

\subsection{Covariance estimation}\label{sec:anaexample3}

Suppose $X$ follows a multivariate normal distribution with zero mean vector and covariance $\Sigma(\theta)$ with elements $\Sigma(\theta)_{jk} = \exp( -\theta  \delta_{jk})$ ($j \neq k$) and $\Sigma(\theta)_{jk} =1$ ($j=k$). The quantity  $\delta_{jk}$ may be regarded as the distance between spatial locations $j$ and $k$, and the case of equally distant points corresponds to covariance estimation for exchangeable variables described in detail by \cite{Cox&Reid04}. The maximum composite likelihood estimator solves
\begin{align*}  
0 = & \sum_{j<k} w_{jk}   \sum_{i=1}^n   U_{jk}(\theta; X_j^{(i)},X_k^{(i)})    \\  
=& \sum_{j<k} w_{jk}   \sum_{i=1}^n   \left[    \dfrac{\Sigma(\theta)_{jk}\{ {X^{(i)}_j}^2+{X^{(i)}_{k}}^2-2 X^{(i)}_jX_k^{(i)}\Sigma(\theta)_{jk}\}}{\{1-\Sigma(\theta)_{jk}^2\}^2} \right] \Sigma(\theta)_{jk} \delta_{jk} \notag\\
&- \sum_{j<k} w_{jk}   \sum_{i=1}^n   \left[\frac{\Sigma(\theta)_{jk}+X^{(i)}_jX_k^{(i)}}{ 1-\Sigma(\theta)_{jk}^2} 
\right]  \Sigma(\theta)_{jk} \delta_{jk} , \notag
\end{align*}
where $U_{jk}(\theta; x_j, x_k)$ is the score of a bivariate normal distribution for the pair $(X_j, X_{k})$.  Figure \ref{fig:anaexample3} shows the solution path of the optimal composition rule $w^\ast_\lambda$  for different values of $\lambda$, and the asymptotic relative efficiency of the estimator $\hat{\theta}_\lambda$ compared to the maximum likelihood estimator. A number of variable pairs ranging from  $m=45$ to $m=1225$   is considered. When $\lambda=0$, the proposed estimator has relatively high asymptotic efficiency.  Interestingly, efficiency stays at around $90\%$ until only a  few sub-likelihoods are left, suggesting that a very small proportion of partial-likelihood components already contains  the majority of the information about $\theta$. In such cases, the proposed selection procedure is useful by reducing the computing burden while retaining satisfactory efficiency for the final estimator.

\begin{figure}[btp]
	\centering
	\quad  \quad       $\theta=0.4$, $m=45$       \quad \quad \quad  $\theta=0.6$, $m=45$ \quad  \quad \quad $\theta=0.6$, $m=1225$\\
	\quad \quad \quad  Number of sub-likelihoods \quad \quad \quad 
	\includegraphics[scale=0.65]{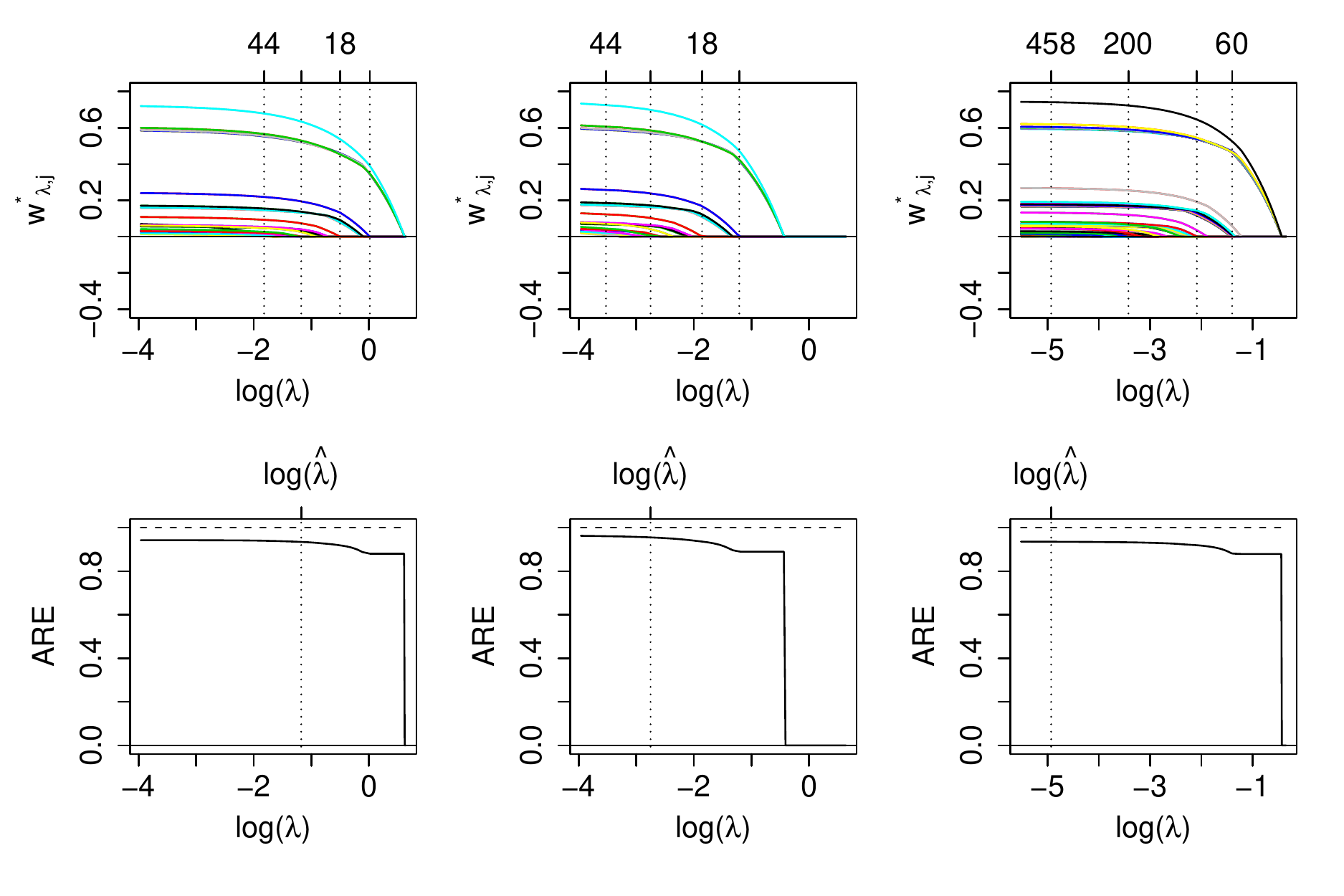}
	\caption{Top Row: Solution paths for the minimizer $w^\ast_\lambda$ of   $d_\lambda(\theta ,w)$ defined in (\ref{eq:criterion2})  at the true parameter for different values of $\lambda$ with corresponding number of sub-likelihoods shown in the top axis.  Bottom Row: Asymptotic relative efficiency (ARE) of the estimator $\hat \theta_\lambda$ compared to the maximum likelihood estimator.   The vertical dashed lines  correspond to $\hat\lambda$ selected as described in \S\ref{Sec:lambda} with $\tau=0.9$. Results correspond to the model $X\sim N_d(0, \Sigma(\theta))$ with   $\Sigma(\theta)_{jk} = \exp\{ -\theta (2|j-k|)^{1/2} \}$. }
	\label{fig:anaexample3}
\end{figure}

\subsection{Location estimation for exchangeable variates} \label{sec:anexample2}

For $X \sim N_m(\theta 1_m, \Sigma)$ with ${\Sigma =  
	(1 -  \rho) I_m + \rho 1_m 1^\top_m }$, $0<\rho<1$,  the marginal scores  $U_j(\theta; X)= X_j-\theta$ ($j=1,\dots, m$) are identically distributed with equal correlation. As $n \rightarrow \infty$ the optimal composition rule converges to
$$
w^\ast_{\lambda,j}= \dfrac{1-\lambda}{\rho(m -1) + 1} I(\lambda<1), \ \ (j=1,\dots, m) 
$$
so the corresponding composite likelihood  estimator is $ \hat \theta_\lambda = \sum_{j=1}^m \bar{X}_j/m$  and is independent of  $\lambda$.  This suggests that  the partial scores are selected randomly in the empirical composition rule $\hat w_\lambda(\theta)$. However,  taking a sufficiently large value for $\lambda$,  so that the sparse  solution containing only a few zero elements still ensures  relatively high statistical efficiency for the corresponding estimator $\hat \theta_\lambda$. To see this, first note that eigenvalue of the score covariance $J(\theta)$ is $\rho(m-1)+1$, whilst the remaining $m-1$ eigenvalues are all equal to $1-\rho$, suggesting that the first score contains relatively large information on $\theta$ compared to the other scores. Furthermore, since  
$
\text{var}(\hat \theta_\lambda ) = \{\rho(m-1)+1\}/(m n)$ the asymptotic relative efficiency of the composite likelihood  with $m< \infty$ compared to that with  $m \rightarrow \infty$ is $\rho m / \{\rho(m-1)+1\}$; this is $0.83$, $0.90$ and $0.98 $, for $m =5, 9$ and $50$, respectively, when $\rho=0.75$. 

\section{Real data example: Spatial covariance estimation for Covid-19 data}\label{sec:realdata}

The methodology is applied to public health data on the Covid-19 epidemic supplied by the Italian Civil Protection Department. The data considered here consists of  $n = 60$ observations on daily new Covid-19 cases from 24 February to 23 April 2020 observed across the $d=90$ Italian provinces corresponding to capital cities in regions or autonomous territories. Data is available as supplementary material to this paper.  Let $X^{(i)}_{j}  = \mu_j^{(i)} + \varepsilon^{(i)}_{j}$ be the  number of new Covid-19 cases observed on day $i$ in province $j$, where $\mu_j^{(i)}$ is a location-specific deterministic trend  and $\sigma^2_j = \text{var}(\varepsilon_j^{(i)})$.
Assume  normally distributed errors $\varepsilon^{(i)}_{1}, \dots, \varepsilon^{(i)}_{m}$ with covariance specified by the gravity model
\begin{equation}\label{eq:Gravity}
\cov(\varepsilon^{(i)}_j, \varepsilon^{(i)}_k) =  \sigma_j \sigma_k \exp\left\{ -\theta {t_{jk}}/{(m_j m_k)} \right\}    \ \   (j, k = 1, \dots, d),
\end{equation}
where $m_j$ denotes the population size (in millions) of the $j$th site and $t_{jk}$ is the distance between sites computed as $t_{jk} = \{(\text{lat}_j -\text{lat}_k)^2 + (\text{lon}_j -\text{lon}_k)^2 \}^{1/2}$, where $\text{lat}_j$ and $\text{lon}_j$ represent latitude and longitude for the $j$th site, respectively. Then each $X^{(i)}= (X^{(1)}_1, \dots, X^{(i)}_d)$ may be regarded as a set of $d$ observations on a Gaussian random field with exponential spatial covariance function given in (\ref{eq:Gravity});  see   \cite{bevilacqua2015comparing} for more details and comparisons on pairwise likelihood estimation for Gaussian random fields. Clearly correlation depends on  population density and potential flows of people across pairs of provinces. To control for these aspects, we included both population size and distance between provinces are included in the  gravity model specified in (\ref{eq:Gravity}). This means that $\theta$ should be interpreted as a covariance parameter for new cases between provinces, conditional  on population size and distance between provinces. Population size and distance are the two main factors often used in formulating distance decay laws, such as gravity models, to represent spatial interactions related to disease diffusion processes.

The main interest is in estimating the covariance parameter $\theta$ to monitor contagion across provinces. To this end, the data are first de-trended and normalized; estimation of the local trends $\hat \mu_j$ are obtained by a Nadaraya-Watson kernel smoother, implemented in the R function \texttt{ksmooth}. The error variance $\sigma^2_j$ are estimated by $\hat \sigma_j^2 = \sum_{i=1}^n (X^{(i)}_j - \hat \mu_j^{(i)})^2/n$. The covariance parameter $\theta$ is subsequently estimated using the pair-wise likelihood approach  in \S\ref{sec:anaexample3} with $\delta_{jk} = t_{jk}/(m_j m_k)$. Table \ref{Tablereal} shows estimates corresponding to a sequence of  decreasing values for $\lambda$. As the number of pair-wise likelihood terms increases,  the estimator $\hat \theta_\lambda$ tends to approach some stable value  and its standard error decreases. For instance, taking $\tau = 0.75$ as defined in \S\ref{Sec:lambda}, the final estimate is $\hat{\theta}_\lambda =  4.25\times 10^{-2}$,  with standard error $1.01 \times 10^{-3}$ corresponding to $58$ pairs of cities selected out of ${90\choose 2} = 4005$ pairs. We also see from the table that both $\hat \theta_\lambda$ and the standard error is converging with only about 58 sub-likelihoods selected. In comparison, the uniform composite likelihood estimator using all pairs of sites with equal weights is $\hat \theta_{\text{unif}} =  6.15\times 10^{-2} $, with standard error $2.42 \times 10^{-3}$.  

\begin{table}[!]
	\caption{Estimates for the spatial covariance parameter for the  Covid-19 data with corresponding standard errors and  number of selected sub-likelihoods (\# sub-likelihoods).}
	\label{Tablereal}
	\footnotesize{
	\begin{tabular}{ cccccccccccc }
		\hline 
		$\lambda$ $(10^5)$      & 8.64 & 8.57 & 8.29 & 8.12 & 7.82 & 7.69 & 7.59 & 7.54 & 7.44 & 7.22 & 7.13\\                
		$\hat \theta_\lambda$ $(10^{-2})$       & 2.58 & 3.53 & 3.58 & 3.90 & 4.03 & 4.12 & 4.23 & 4.25 & 4.25 & 4.25 & 4.25\\
		$SE$ $(10^{-2})$      & 4.50 & 2.34 & 2.17 & 1.26 & 0.84 & 0.69 & 0.47 & 0.46 & 0.11 & 0.10 &  0.08 \\
		$\#$ sub0likelihoods      & 40 & 42 & 44 & 46 &48 &50 & 52& 54& 56 & 58 & 60\\ \hline 
	\end{tabular}
}
\end{table}

\section{Conclusion and final remarks}
\label{sec:discussion}

Composite likelihood inference plays an important role as a remedy to the drawbacks of traditional likelihood approaches with advantages in terms of computing and modeling flexibility. Nevertheless,    a universal procedure to construct composite likelihoods that is statistically justified and fast to execute  does not seem to exist \citep{Lindsay&al11}. Motivated by this gap in the literature, this paper introduces a selection methodology resulting in composite likelihood estimating equations with good statistical properties. The selected equations are sparse for sufficiently large $\lambda$, meaning that they contain only the most informative sub-likelihood score terms. This sparsity-promoting mechanism is found to be   useful in common situations where  the sub-likelihood scores are heterogeneous in terms of their information or when the ideal $O_F$-optimal score is difficult to obtain. Remarkably, the sparse score is shown to approximate $O_F$-optimal score in large samples under reasonable conditions; see Theorem \ref{prop:proposition2} and Equation \ref{eq:equivalence}.

For implementation, we proposed a selection criterion to choose $\lambda$ which perform well in the examples, instead of providing a universal approach. In practice, it could be feasible to use any alternative criterions to choose $\lambda$, according to the realization of the problems. For example, when the full score covariance is not available due to computational burden, one may consider to use the upper bound provided in Section \ref{Sec:lambda}, or to choose $\lambda$ up to some given level of information gain, defined by the ratio of the smallest eigenvalue to the trace of the current selected score covariance, which is decreasing with $\lambda$ by min-max theorem in linear algebra. As another idea, we note that by the Karush-Kuhn-Tucker condition of quadratic optimization, $\lambda$ represents the norm of the estimated covariance between the current selected sub-score $U_j(\theta; X)$ and the residual $\{U(\theta,w;X)-U_j(\theta;X)\}$. One can choose $\lambda$ such that the covariance is smaller than some pre-fixed value.

Building on the recent success of shrinkage methods for the full likelihood, many  works have proposed the use of sparsity-inducing penalties in the composite likelihood framework; e.g., see \cite{bradic2011penalized, xue2012nonconcave,  gao2017data}. However, the spirit of our approach is entirely different from these methods, since our penalty focuses on entire sub-likelihood functions rather than on elements of $\theta$. In contrast to the above approaches, our penalization strategy has the advantage of retaining asymptotically unbiased estimating equations, thus leading to desirable asymptotic properties of the related parameter estimator.

A number of  developments of the present study may be pursued in the future from either theoretical or applied viewpoints. Although the current paper focuses on the case where $p$ is finite, penalties able to deal with situations where both $m$ and $p$ are allowed to grow with $n$ may be useful for the analysis of high-dimensional data. Implementations of the convex  efficiency criterion (\ref{eq:criterion}) beyond the current i.i.d. setting would be another useful future research direction. For example, this  would be valuable for the analysis of spatial or spatio-temporal data, where often the overwhelming number of sub-likelihoods poses a challenge to traditional composite likelihood methods.

\section*{Acknowledgement} 

The authors acknowledge the financial support of the Italian Ministry of University and Research (MIUR) -- Research Project of National Interest (PRIN) grant  2017TA7TYC.

\section*{Appendix: Proofs  }

For convenience, we use $U_j(\theta;X^{(i)})$ for $U_j^{(i)}$ and let $M^{(i)}$ be the $p\times m$ matrix collecting $U^{(i)}_j$ for $j=1,\dots,m$. Let $M^{(i)}_{\hat \varepsilon}$ be the sub-matrix of $M^{(i)}$ with columns indexed by the set $\hat\varepsilon$ and denote $U^{(i)} = U\{\theta;\hat w_\lambda(\theta),X^{(i)}\} = M^{(i)} \hat w_\lambda(\theta)$. 

\begin{proof}[Proof of Theorem \ref{prop:uniqueness}]
	We first note that for any $\theta \in \Theta$, $\hat d_\lambda(\theta,w)$ is lower bounded, thus the minimizer exists. This is implied by taking the eigen decomposition of the real Hermitian matrix $\hat J(\theta)$ and then re-organize $\hat d_\lambda(\theta,w)$ as a summation of perfect square terms corresponding to nonzero eigenvalues, a non-negative first order term corresponding to zero eigenvalues and a constant.   We also note that $U^{(i)}$ is unique due to the strict convexity of the first term of $\hat d_\lambda (\theta,w)$ with respect to $U^{(i)}$, the convexity of the rest of the terms with respect to $w$, and the linearity of $U^{(i)}$ with respect to $w$. By the Karush-Kuhn-Tucker conditions for quadratic optimization, the solution must satisfy
	\begin{equation} \label{eq:kkt}
	\frac{1}{n}\sum_{i=1}^n {U_j^{(i)}}^\top U^{(i)} - \frac{1}{n}\sum_{i=1}^n {U_j^{(i)}}^\top U_j^{(i)} + \lambda \gamma_j =0, \quad\quad\text{for $j= 1,\dots,m$},
	\end{equation}
	where $\gamma_j = \mathrm{sign}(w_j)$ if $w_j\neq 0$ and $\gamma_j \in [-1,1]$ if $w_j=0$. This implies that $\hat \varepsilon$ defined in \ref{eq:epsilonhat} is unique. Note that the rank of $\hat J_{\hat \varepsilon} \equiv n^{-1}\sum_{i=1}^n {M^{(i)}_{\hat\varepsilon}}^\top M^{(i)}_{\hat\varepsilon}$ is at most $\min(m,np)$. Next we show that $\hat J_{\hat \varepsilon}$ has full rank. Otherwise, by the rank equality of the Gram matrix, there exists a subset $\tilde \varepsilon \subseteq \hat \varepsilon$, $|\tilde \varepsilon| \leq \min(m,np+1)$ and some $k\in\tilde \varepsilon$, such that $({U_k^{(1)}}^\top, \dots, {U_k^{(n)}}^\top)$ can be written as a linear combination of $({U_j^{(1)}}^\top, \dots, {U_j^{(n)}}^\top)$, for $j\in\tilde \varepsilon$ and $j\neq k$. Together with the Karush-Kuhn-Tucker condition, there exist  constants $a_j$, $j\in\tilde \varepsilon$ and $j\neq k$ (with $a_j \neq 0$ for some $j$) such that $U_k^{(i)}= \sum_{j\in\tilde\varepsilon, j\neq k} a_j U_j^{(i)}$, for all $i=1,\dots,n$, and 
	\begin{equation*} 
	\frac{1}{n} \sum_{i=1}^n {U_k^{(i)}}^\top U_k^{(i)} + \lambda \gamma_k  = \sum_{j\in\tilde\varepsilon, j\neq k} a_j \frac{1}{n} \sum_{i=1}^n {U_j^{(i)}}^\top U_j^{(i)} + \lambda \gamma_j.
	\end{equation*}
	This represents a linear system with $(np+1)$ equations but only $|\tilde \varepsilon|-1$ degrees of freedom, meaning that the rank of the $(np+1) \times |\tilde \varepsilon|$ matrix generated by columns $({U_j^{(1)}}^\top,\dots,{U_j^{(n)}}^\top, (1/n)\sum_{i=1}^n {U_j^{(i)}}^\top U_j^{(i)} + \lambda \gamma_j )$, $j \in \tilde\varepsilon$ is smaller or equal to $|\tilde \varepsilon|-1$. Since $|\tilde \varepsilon|\leq np+1$, we have that the $|\tilde\varepsilon|$ columns are linearly dependent. Under Condition \ref{cond1}, this event has zero probability, which is a contradiction. The statement in the theorem then follows by solving the Karush-Kuhn-Tucker equations in (\ref{eq:kkt}).
\end{proof}

\begin{lemma} \label{lemma1}
	Under Conditions \ref{cond1} and \ref{cond2}, for any $\lambda>0$,  $\sup_{\theta \in \Theta} \Vert \hat d_\lambda\{\theta,\hat w_\lambda(\theta)\} - d_\lambda\{\theta,\hat w_\lambda(\theta)\} \Vert_1  \rightarrow 0$ in probability, as $n \rightarrow \infty$.
\end{lemma}

\begin{proof}[Proof of Lemma \ref{lemma1}]
   	By definition, it suffices to show that for all $j,k\geq 1$, $sup_{\theta_\in \Theta} | \hat J(\theta)_{jk} - J(\theta)_{jk} |\rightarrow 0 $ in probability as $n \rightarrow \infty$, and that $\|\hat w_\lambda(\theta)\|_1$ is uniformly bouned with probability tending to one. The first part is ensured by Condition \ref{cond2}. For the second part, by Theorem \ref{prop:uniqueness}, it suffices to show that $\hat J_{\hat \epsilon}(\theta)$ and $\hat J_{\hat \epsilon}(\theta)^{-1}$ are uniformly bounded entry-wise with probability tending to 1, which is guarenteed  by the uniform convergence of $\hat J(\theta)$ in probability, the boundedness of each element of $J(\theta)$ and the invertibility of $\hat J_{\hat \epsilon}(\theta)$ accoridng to the min-max theorem in linear algebra. 
\end{proof}

\begin{proof}[Proof of Theorem \ref{prop:proposition2}]
	Recall  that $\hat w_\lambda(\theta) = \mathrm{argmin}_w \hat d_\lambda(\theta,w)$ and $  w_\lambda(\theta) = \mathrm{argmin}_w   d_\lambda(\theta,w)$. Let $\xi$ be the smallest eigenvalue of $J(\theta)$. By definition, we have
\begin{align*}
	    &\sup_\theta\left[\frac{1}{2}\xi \| \hat w_\lambda(\theta) - w_\lambda(\theta) \|_2^2 \right]\\
	\leq & \sup_\theta\left[\frac{1}{2} \{\hat w_\lambda(\theta) - w_\lambda(\theta)\}^T J(\theta) \{\hat w_\lambda(\theta) - w_\lambda(\theta)\}\right]\\
	\leq & \sup_\theta\left[|d\{ \theta,\hat w_\lambda(\theta)  \} - d\{ \theta,w_\lambda(\theta)  \}|\right] \\
	\leq  & \sup_\theta\left[|d\{ \theta,\hat w_\lambda(\theta)  \} - \hat d\{ \theta,\hat w_\lambda(\theta)  \}|\right] + \sup_\theta\left[|\hat d\{ \theta,  w_\lambda(\theta)  \} - d\{ \theta,w_\lambda(\theta)  \}|\right]
\end{align*}	
 where the second inequality is due to the Karush-Kuhn-Tucker conditions of quadratic optimization, and the second last inequality is due to that $\hat w_\lambda(\theta)$ and $w_\lambda(\theta)$ are the corresponding minimizers.	By Lemma \ref{lemma1}, the first term of the last inequality converges to zero in probability, and the same holds for the second term. Under Condition \ref{cond2}, $\xi>0$. Since $m$ is fixed, it concludes the proof. 
\end{proof}

\begin{lemma} \label{lemma2}
	Under Conditions \ref{cond1} and \ref{cond2}, $w_\lambda(\theta)$ is continuous with respect to both $\lambda$ and $\theta$ on $\lambda \geq 0$ and $\theta \in \Theta$. 
\end{lemma}

\begin{proof}[Proof of Lemma \ref{lemma2}]  
	For simplicity, here we show the continuity of $w_\lambda(\theta)$ with respect to $\theta$. The proof for continuity with respect to $\lambda$ is the same and thus omitted. For any $c>0$ and $\theta_1 \in \Theta$, it suffices to show that there exist some $\delta >0$, such that $\| \theta-\theta_1  \| _1 < \delta$ implies $\|w_\lambda(\theta)-w_\lambda(\theta_1)\| < c$. To find $\delta$, recall that $w_\lambda(\theta )$ is the minimizer of $d_\lambda(\theta,w)$ defined in \ref{eq:criterion2}. Under Condition \ref{cond2}, $d_\lambda(\theta,w)$ is strictly convex with respect to $w$. Thus, there exists $c_1 = \inf_{ \{w:\| w-w_\lambda(\theta_1) \|_1=c\} } d_\lambda(\theta_1,w)  > d_\lambda\{\theta_1,w_\lambda(\theta_1)\}$.   Moreover, $d_\lambda(\theta,w)$ is uniformly continuous on the closed domain $\{w\in \mathbb R^m: \| w-w_\lambda(\theta_1) \|_1 \leq c  \}\times \{ \theta \in \mathbb R^p: \| \theta -\theta_1 \|_1 \leq \delta_1 \}$ for some $\delta_1>0$. Thus we can find $\delta \in (0,\delta_1)$ such that for any $\{ \theta: \|\theta-\theta_1 \|_1 < \delta \}$ and $\{w : \| w-w_\lambda(\theta_1) \|_1 \leq c \}$, $\| d_\lambda(\theta,w) - d_\lambda(\theta_1,w)  \|_1 < \{c_1 - d_\lambda(\theta_1,w) \}/2$. This implies that when $\|\theta-\theta_1 \|_1 < \delta$, $d_\lambda(\theta,w) > d_\lambda\{\theta,w_\lambda(\theta_1)\}$ for all $\{w:\| w-w_\lambda(\theta_1) \|_1=c\}$. Since $d_\lambda(\theta,w)$ is strictly convex, we have $\|w_\lambda(\theta)-w_\lambda(\theta_1)\| < c$.
\end{proof}

\begin{proof}[Proof of Theorem \ref{thm:consistent}]
	Note that $\hat \theta_\lambda$ and $\theta^\ast$ are the solutions of the estimating equations  $U(\theta, w;  X^{(1)}, \dots, X^{(n)})/n= 0$ and $E\{U(\theta,w;X)\}=0$, with $w$ replaced by $\hat w _\lambda(\tilde \theta)$ and $w_\lambda^\ast$, respectively.  
	By Theorem \ref{prop:proposition2} and Lemma \ref{lemma2},  we have $   \| \hat w_\lambda(\tilde \theta) - w_\lambda^\ast \|_1 \to 0$ in probability, as $n \to \infty$. Under Condition \ref{cond2}, $E\{U_j(\theta,X)\}$ is bounded. Moreover, under Condition \ref{cond3}, each sub-likelihood score $\sum_{i=1}^n U_j(\theta;X^{(i)})/n \to E\{ U_j(\theta;X) \}$ uniformly with probability tending to one. Since  $ U\{\theta, \hat w _\lambda(\tilde \theta); X^{(1)},\dots,X^{(n)} \}/n$ and $E \{U(\theta, w_\lambda^\ast; X)\}$ are the product of $\hat w_\lambda(\tilde \theta)$, $w_\lambda^\ast$ and the sub-likelihood scores, we have
	\begin{align}\label{eq:unif}
	\sup_{\theta \in  \Theta} \left\| \frac{1}{n} U\{\theta, \hat w _\lambda(\tilde \theta); X^{(1)},\dots,X^{(n)} \}  -   E \{U(\theta, w_\lambda^\ast; X)\} \right \|_1 \to 0,
	\end{align} in probability as $n \to \infty$.

Moreover, by Condition \ref{cond3}, $\inf_{\{ \theta: \| \theta-\theta^\ast \|_1\geq c\} } \| E\{U(\theta,w_\lambda^\ast;X)\} \|_1 >0$ for any constant $c>0$. Thus, for any $c>0$, there exists a $\delta>0$ such that the event $ \| \hat \theta_\lambda-\theta^\ast \| \|_1 \geq c$ implies the event $\| E\{U( \hat \theta_\lambda, w_\lambda^\ast;X)\} \|_1 >\delta$. We have
\begin{align*}
	 & \Pr\{\| \hat \theta_\lambda-\theta^\ast \| \|_1  \geq  c \} \\ \leq & \Pr\{\| E\{U( \hat \theta_\lambda, w_\lambda^\ast;X)\} \|_1 >\delta \} \\ 
	 = & \Pr\{\| E\{U( \hat \theta_\lambda, w_\lambda^\ast;X)\} -  \frac{1}{n} U\{\hat \theta_\lambda, \hat w _\lambda(\tilde \theta); X^{(1)},\dots,X^{(n)} \} \|_1 >\delta \}\\   \rightarrow & 0,
\end{align*} 	
as $n \to \infty$, where the equality is due to that $U\{\hat \theta_\lambda, \hat w _\lambda(\tilde \theta); X^{(1)},\dots,X^{(n)} \} =0 $ and the last line is  implied by (\ref{eq:unif}).  This concludes the proof.
\end{proof}


\begin{proof}[Proof of Theorem \ref{thm:normality}]
	  Note that $U\{ \hat \theta_\lambda , \hat w_\lambda(\tilde\theta );  X^{(1)}, \dots, X^{(n)}\}/n= 0$, and that $\| \hat \theta_\lambda - \theta^\ast \|_1 \to 0$ in probability as $n \to \infty$. By Condition  \ref{cond4} and applying the law of large number to the remainder, we obtain the following expansion of $j$th sub-likelihood score at $\theta^\ast$,
	  \begin{align*}
	  & \dfrac{1}{n}\sum_{i=1}^n U_j(\hat \theta_\lambda; X^{(i)})  -  \dfrac{1}{n} \sum_{i=1}^n U_j(\theta^\ast; X^{(i)}) \\  = &  \dfrac{1}{n} \sum_{i=1}^n \nabla U_j(\theta^\ast; X^{(i)}) (\hat \theta_\lambda -\theta^\ast) + o_p(\| \hat \theta_\lambda - \theta^\ast \|_1 1_p) ,
	  \end{align*}
	  where $1_p$ is the $p$-dimensional vector with elements equal to one. Note that under Condition \ref{cond2}, $\|\hat w _\lambda(\theta)\|_1$ is uniformly bounded (also see the proof of Lemma \ref{lemma1}). Taking the entry-wise product of the empirical composition rule and the sub-liekilhood scores implies 
	  \begin{align}  
	  	&\sqrt n \frac{1}{n} U\{ \theta^\ast , \hat w_\lambda( \tilde \theta );  X^{(1)}, \dots, X^{(n)}\} \notag \\ = &- \sum_{j=1}^m \hat w_\lambda(\tilde \theta)_j \frac{1}{n} \sum_{i=1}^n \nabla U_j(\theta^\ast; X^{(i)})\left\{ \sqrt n  (\hat \theta_\lambda - \theta^\ast)\right\} +  o_p(\sqrt n \| \hat \theta_\lambda - \theta^\ast \|_1 1_p)  , \label{eq:clt}
	  \end{align}
where $\hat w_\lambda(\tilde \theta)_j$ is the $j$th element of $\hat w_\lambda(\tilde \theta)$. Note that by Theorem \ref{prop:proposition2} and  Lemma \ref{lemma2}, $\| \hat w_\lambda(\tilde \theta) - w_\lambda^\ast \|_1 \to 0$ in probability.   Under Condition \ref{cond4}, by the Central Limit Theorem and Slustky's Theorem, the left-hand side of (\ref{eq:clt}) converges in distribution to a multivariate normal random vector with mean zero and covariance $K(\theta^\ast,w_\lambda^\ast) = \mathrm{cov}\{U( \theta^\ast,w_\lambda^\ast ;X)\}$ defined in (\ref{eq:Godambe}). The $p \times p$ matrix $-\sum_{j=1}^m \hat w_\lambda(\tilde \theta)_j   \sum_{i=1}^n \nabla U_j(\theta^\ast; X^{(i)})/n$ in right hand side of (\ref{eq:clt}) converges in probability to $H(\theta^\ast ,w_\lambda^\ast)$ defined in (\ref{eq:Godambe}) by the Law of Large Numbers and Slustky's Theorem. The invertibility of $H(\theta^\ast,w_\lambda^\ast)$ implies that $\hat \theta_\lambda$ is root-$n$ consistent. Re-organizing (\ref{eq:clt}) implies the desired  result. 
\end{proof}


\bibhang=1.7pc
\bibsep=2pt
\fontsize{9}{14pt plus.8pt minus .6pt}\selectfont
\renewcommand\bibname{\large \bf References}
\expandafter\ifx\csname
natexlab\endcsname\relax\def\natexlab#1{#1}\fi
\expandafter\ifx\csname url\endcsname\relax
  \def\url#1{\texttt{#1}}\fi
\expandafter\ifx\csname urlprefix\endcsname\relax\def\urlprefix{URL}\fi

  \bibliographystyle{chicago}      
  \bibliography{biblio}   

\vskip .65cm
\noindent
Zhendong Huang, School of Mathematics and Statistics, University of Melbourne, Peter Hall Building,
3010 Parkville, Australia
\vskip 2pt
\noindent
E-mail: (huang.z@unimelb.edu.au)
\vskip 2pt

\noindent
Davide Ferrari, Faculty of Economics and Management, University of Bolzano, Piazza Universit\`{a} 1, Bolzano,
39100, Italy
\vskip 2pt
\noindent
E-mail: (davferrari@unibz.it)

\end{document}